\documentclass[twoside,12pt,leqno]{amsproc}
\usepackage{amssymb,latexsym,enumerate,tikz}
\usetikzlibrary{patterns}
\usepackage[pagebackref]{hyperref}
\usepackage{amsrefs}    
\usepackage{etoolbox}   
\hypersetup{citecolor=purple, linkcolor=blue, colorlinks=true}
\usepackage{bookmark}
\usepackage{mathtools}
\usepackage{nicefrac}
\numberwithin{table}{section}

\theoremstyle{plain}
\newtheorem{theorem}{Theorem}[section]
\newtheorem{lemma}[theorem]{Lemma}
\newtheorem{proposition}[theorem]{Proposition}

\newtheorem{construction}[theorem]{Construction}

\theoremstyle{definition} 
\newtheorem{definition}[theorem]{Definition}
\newtheorem{example}[theorem]{Example}
\newtheorem{remark}[theorem]{Remark}

\renewcommand{\geq}{\geqslant}
\renewcommand{\leq}{\leqslant}
\renewcommand{\ge}{\geqslant}
\renewcommand{\le}{\leqslant}

\newcommand{\A}{\mathcal{A}}

\newcommand{\B}{\mathcal{B}}

\newcommand{\F}{\mathbb{F}}

\newcommand{\wt}{\mathrm{wt}}
\def\code{{\sf Code}}
\def\basis{{\sf Basis}}

\makeatletter        
\def\@adminfootnotes{%
  \let\@makefnmark\relax  \let\@thefnmark\relax
  \ifx\@empty\@date\else \@footnotetext{\@setdate}\fi
  \ifx\@empty\@subjclass\else \@footnotetext{\@setsubjclass}\fi
  \ifx\@empty\@keywords\else \@footnotetext{\@setkeywords}\fi
  \ifx\@empty\thankses\else \@footnotetext{%
    \def\par{\let\par\@par}\@setthanks}%
  \fi}\makeatother   

\begin{document}

\hyphenation{}


\thanks{
  2010 Math Subject Classification: 94B65, 94B05, 94B15.\hfill
 Date: \today.
}
\title[Sequences of linear Codes with rate times distance unbounded]{Sequences of linear codes where the\\ rate times distance grows rapidly}
\author[Faezeh Alizadeh, S. P. Glasby, Cheryl E. Praeger]{Faezeh Alizadeh, S.\,P. Glasby and Cheryl E. Praeger}

\address[Alizadeh]{
Shahid Rajaee Teacher Training University,
Tehran, Iran.\newline Email: {\tt Faezeh.Alizadeh2@gmail.com}
}
\address[Glasby, Praeger]{
Centre for the Mathematics of Symmetry and Computation,
University of Western Australia,
35 Stirling Highway,
Perth 6009, Australia.\newline
 Email: {\tt Stephen.Glasby@uwa.edu.au; Cheryl.Praeger@uwa.edu.au}
}
\begin{abstract}
For a linear code $C$ of length $n$ with dimension $k$ and minimum distance~$d$, it is desirable that the quantity $kd/n$ is large.
Given an arbitrary field $\F$, we introduce a novel, but elementary, construction that produces a recursively defined sequence of $\F$-linear codes $C_1,C_2, C_3, \dots$
with parameters $[n_i, k_i, d_i]$ such that $k_id_i/n_i$ grows quickly in the sense that $k_i d_i/n_i>\sqrt{k_i}-1>2i-1$.
Another example of quick growth comes from a certain subsequence of Reed-Muller codes. Here the field is $\F=\F_2$ and $k_i d_i/n_i$ is asymptotic to $3n_i^{c}/\sqrt{\pi\log_2(n_i)}$ where $c=\log_2(3/2)\approx 0.585$.
\end{abstract}
\maketitle
\section{Introduction}
Let $\F$ be an arbitrary field and let $\F^{n}$ denote the space of all vectors of length $n$ over~$\F$.
An  $\F$-linear code $C$ of length $n$ is  a subspace of the vector space $\F^n$, see~\cite{Huffman}*{Section 1.2}, and elements of $C$ are called \emph{codewords}.
Denote the dimension of $C$ by~$k$.
Given a codeword  $x \in C$, the \emph{(Hamming) weight} of $x$, denoted
by $\wt(x)$, is the number of non-zero coordinates in $x$.
The \emph{(minimum) distance} of $C$, denoted $d$ or $d(C)$, is $d(C) = \min \{d(x, y) \ | \ x, y\in C, x \ne y \}$ where $d(x,y)=\wt (x-y)$ is the number of coordinate entries where $x$ and $y$ differ.
An $\F$-linear code of length $n$, dimension $k$ and distance $d$ is abbreviated
an $[ n, k, d ]$-linear code.
The \emph{(information) rate} of an $[n, k, d]$-linear code $C$ is $R(C)= k/n$.
We henceforth consider only linear codes. 

Let $C_1,C_2,C_3,\dots$ be a sequence of $\F$-linear codes where $C_i$ has parameters $[n_i,k_i,d_i]$.
We are interested in the question: How quickly can $k_id_i/n_i$ approach infinity? It is easiest to understand the growth as a function
of one variable, such as $k_i$, $d_i$ or $n_i$. Since $k_id_i/n_i\le\min\{k_i,d_i\}\le n_i$ such functions will be at most linear.
Asymptotically good code sequences are ones which satisfy both $k_i/n_i\ge c$ and $d_i/n_i\ge c$ for some constant $c>0$, see~\cite{EKL}*{(1.1)}
and~\cite{ShiWuSole}, and hence  these code sequences achieve linear growth in $n_i$ since $k_id_i/n_i\geq c^2n_i$.
Such code sequences have been known
to exist since~\cite{CPW}, but they are not well understood, and the existence of asymptotically good \emph{cyclic} codes has been a long standing open problem~\cites{MW,ShiWuSole}
related to the uncertainty principle~\cite{EKL}. Other code sequences with reasonably good growth rates include the quadratic residue codes, see \cite{Huffman}*{Section 6.6}. These have parameters $[p,(p+1)/2,d]$ where $p$ is a prime, and $d\ge \sqrt{p}$ by~\cite{MacS}*{p.\,483, Theorem~1} so $kd/n$ is at least
$O(\sqrt{n})$. In Example~\ref{E:RM}, we describe a subsequence of Reed-Muller codes over $\F_2$ (\emph{c.f.}~\cite{L}) with faster growth
\emph{viz.} $kd/n=O(n^{\log_2(3/2)})$. On the other hand, we might ask about `non-examples', that is, code sequences for which $k_id_i/n_i$ does not grow. There are such examples even among well-known code families. For example, for any fixed prime-power $q$, there is an $\F_q$-linear Reed-Solomon code~\cite{QBC} with parameters $[n,k,n-k+1]$ provided $q\ge n$. The family of all such codes has $kd/n$ bounded as $kd/n\le\min\{k,d\}\le n\le q$.

We want $k$ to be large (relative to $n$) to transmit a lot of information, and $d$ to be large to correct many errors, see Section~\ref{S:FM} for more details.
For linear codes over a finite field, the Singleton bound ($d \le n-k+1$) exhibits the tension between $k$ and $d$, and
Maximum Distance Separable (MDS) codes have $d=n-k+1$, see \cite{San}. 
In Theorem~\ref{main} we exhibit an explicit construction for a sequence  $C_1, C_2,C_3, \dots$ of $\F$-linear codes, where $\F$ is a fixed but arbitrary field, and where $C_i$ has parameters $[n_i, k_i, d_i]$ and $k_id_i/n_i=O(\sqrt{k_i})$. 



\begin{theorem}\label{main}
For each field $\F$, there exists an infinite sequence $C_1,C_2,C_3,\dots$ of $\,\F$-linear codes such that $C_i$ has parameters $[n_i, k_i, d_i]$ that satisfy \[\displaystyle\frac{k_id_i}{n_i}=2i >\sqrt{k_i}-1>2i-1\qquad\textup{ for $i\ge1$}.\]
\end{theorem}

 The codes $C_1,C_2,C_3,\dots$ in Theorem~\ref{main} are described in Definitions~\ref{defff} and~\ref{def:Fi}. They involve a novel construction described in Section~\ref{sec:bdd}, and their parameters $[n_i,k_i,d_i]$ are determined in Theorem~\ref{corf}. The codes are $\F$-linear, where $\F$ is an arbitrary field, and their limiting properties are similar to the well-known Reed-Muller~codes, see Example~\ref{E:RM}.

%

New codes can also be constructed from old codes via ``direct sums'' and ``repetition''. Both of these constructions fix the quantity $kd/n$, for
suppose that $C$ is an $\F$-linear code with parameters $[n,k,d]$, and $s$ is a positive integer.
The direct sum code $\bigoplus_{i=1}^s C\le\F^{ns}$ has parameters $[ns,ks,d]$, and the quantity $(ks)d/(ns)=kd/n$ is independent of $s$. Similarly, the repetition
code $\textup{diag}(\bigoplus_{i=1}^s C)=\{(u,\dots,u)\mid u\in C\}$ has parameters $[ns,k,ds]$, and again $k(ds)/(ns)=kd/n$ is independent of $s$.
By contrast, repeated application of the construction introduced to prove Theorem~\ref{corf} allows $kd/n$ to grow without bound.

\section{Further remarks}\label{S:FM}

We represent the vectors in $\F^n$ by $n$-dimensional column vectors over $\F$. Let $C$ be an $[n,k,d]$-linear code in $\F^n$ and let $B= (a_1, \dots , a_k)$ be an  ordered basis of $C$. Let $H=[a_1, \dots ,a_k]$ be the $n\times k$ matrix formed by this basis. Then the  code $C$ is equal to the set $\{ Hx \mid x\in \F^k\}$, and we regard an element $x\in \F^k$ as an information vector which is encoded as a codeword $Hx$. As $H$ has rank $k$, by permuting the coordinates of $\F^n$ if necessary, we will assume that the 
$k \times k$ sub-matrix formed by the first $k$ rows of $H$ is invertible. Thus the information vector $x$ is uniquely determined by the first $k$ entries of the associated codeword $Hx$ and the last $n-k$ entries of $Hx$ may be regarded as redundancy added to enable transmission errors to be corrected.      

The proportion $k/n$ is called the \emph{rate} $R(C)$ of $C$, and is a measure of the efficiency of $C$ in encoding information. The third parameter, the minimum distance $d$ of $C$ is directly related to the error correction capability of $C$: if strictly less than $d/2$ entries in a codeword are changed during transmission then the received $n$-tuple is closer (in the Hamming metric) to the original codeword transmitted than to any other codeword, and hence these `errors’ can be corrected and the correct codeword determined.  A high value of $d$ relative to $n$ denotes high reliability in transmitting information correctly.    

Both $k/n$ and the \emph{relative distance} $d/n$ are less than $1$: the larger the rate $k/n$ the more efficient is the code, while the larger the relative distance $d/n$ the more reliable is the code.  The question addressed in this paper is: How quickly can the quantity $k_id_i/n_i$ approach infinity? Our construction demonstrates that this quantity $kd/n$ can grow as fast as $\sqrt{k}$ for codes over any field $\F$.  

\begin{example}\label{E:RM}
A Reed-Muller code $\textup{RM}(m,r)$ is an $\F_2$-linear code, with parameters $[2^m,\sum_{j=0}^r\binom{m}{j},2^{m-r}]$ by~\cite{Huffman}*{Theorem 1.10.1}.
Thus $k_rd_r/n_r=2^{-r}\sum_{j=0}^r\binom{m}{j}$ for $\textup{RM}(m,r)$. The subsequence $(\textup{RM}(2r+1,r))_{r\ge1}$
has $\dim(\textup{RM}(2r+1,r))=k_r=\sum_{j=0}^r\binom{2r+1}{j}=2^{2r}$ and parameters $[2^{2r+1},2^{2r},2^{r+1}]$. Therefore $k_rd_r/n_r=\sqrt{k_r}=O(\sqrt{n_r})$.  It turns out that the Reed Muller codes $\textup{RM}(2r+1,r)$, and
the codes $C_r$ in the Theorem~\ref{corf} have $k_rd_r/n_r=O(\sqrt{k_r})$, and the quadratic residue codes have $k_rd_r/n_r$ at least $\sqrt{k_r}$.

A faster growing subsequence is $(\textup{RM}(m,\lfloor \frac{m}{3}\rfloor +1))_{m\ge1}$. It follows from~\cite{GP}*{Theorem~2} that the parameters
$[n_m,k_m,d_m]$ of $\textup{RM}(m,\lfloor \frac{m}{3}\rfloor +1)$ satisfy $k_m d_m/n_m$ is asymptotic to
\[
\frac{3}{\sqrt{\pi m}}\left(\frac{3}{2}\right)^m=O\left(\frac{n_m^{\log_2(3/2)}}{\sqrt{\log_2(n_m)}}\right)\qquad\textup{ where $n_m=2^m$ and $\log_2\left(\frac{3}{2}\right)\approx0.585$}.
\]
\end{example}

\section{Bounded linear codes}\label{sec:bdd}

In this section we present, and investigate, a construction which, given a linear code as input, produces a linear code with strictly larger length, dimension and minimum distance. We will be concerned with the extent to which the construction preserves the property of being bounded, which is defined as follows. 

\begin{definition}\label{defadm}{\rm
Let  $u$ be a positive integer, $C$ an $\F$-linear $[n,k,d]$-code, and $B=(a_1,\dots,a_k)$ an ordered basis for $C$. Then $C$ is said to be \emph{$u$-bounded relative to $B$} if the following three conditions hold:
\begin{enumerate}[{\rm (i)}]
\item
$\wt(a_j) = u$ for each $j$;
\item
$\wt(\sum_{j=1}^ka_j$)=$d$; and
\item
$u\ge d(1+k^{-1})$.
\end{enumerate}
We say that $C$ is \emph{bounded} if there exist $u$ and $B$ such that $C$ is $u$-bounded  relative to $B$.
}
\end{definition} 
\subsection{The construction}\label{sub:con}

Here we present the construction applied to an arbitrary linear code, and examine when the boundedness property is preserved.

\begin{construction}\label{conadm} 
We have the following input an output.\newline
{\sc Input:}\quad Let $C$ an $\F$-linear $[n,k,d]$-code in $\F^n$ with ordered basis $B=(a_1,\dots,a_k)$.
\newline
{\sc Output:}\quad The $\F$-linear code $\code(B)$ in $\F^{n(k+1)}$ which is the $\F$-linear  span of the sequence $\basis(B)=(a_1',\dots a_{k+1}')$, where 
\[
a_{1}' =
\begin{bmatrix}
0\\
a_{1} \\ 
a_{2} \\
\vdots \\
a_{k}
\end{bmatrix}, \ 
a_{2}' =
\begin{bmatrix}
a_{k} \\
0 \\
a_{1} \\
\vdots \\
a_{k-1}
\end{bmatrix},\ 
\dots, \ 
a_{k+1}' =
\begin{bmatrix}
a_{1} \\ 
a_{2} \\
\vdots \\
a_{k} \\
0 \\
\end{bmatrix},
\]
and $0$ denotes the zero vector in $\F^n$.       
\end{construction}

We show that the output $\basis(B)$  from Construction~\ref{conadm} is linearly independent, determine the parameters of the code $\code(B)$, and find conditions under which $\code(B)$ is bounded relative to $\basis(B)$. 

\begin{proposition}\label{propadm}
Let $\F, C, n,k,d,B$ be as in Construction~$\ref{conadm}$. 
\begin{enumerate}[{\rm (a)}]
    \item  Then $\basis(B)$ is an ordered basis for $\code(B)$, and $\code(B)$ is an $\F$-linear  $[n',k',d']$-code, where
\[
n'=(k+1)n,\quad k'=k+1  \quad \text{and $d'$ satisfies} \quad d'\geq kd.
\]
    \item Further, if $C$  is $u$-bounded relative to $B$, then $d'=(k+1)d$, and, setting $u'=ku$, $\code(B)$ is $u'$-bounded relative to $\basis(B)$ if and only if $u\ge d(1+2k^{-1})$.
\end{enumerate}
\end{proposition}
\begin{proof}
(a) By Construction~\ref{conadm}, $\code(B)$ has length $n'=(k+1)n$. Further, the fact that the $a_j$ are linearly independent implies that the $a_j'$ are linearly independent, so $\basis(B)$ is an ordered basis for $\code(B)$, and therefore $\code(B)$ has dimension $k'=k+1$.
It remains to bound the minimum distance $d'$.
%
An arbitrary non-zero codeword $w$ in $\code(B)$ has the form $w=\sum_{j=1}^{k+1} c_ja_j'$ for elements $c_j\in \F$, not all zero. 
We will prove that $\wt(w)\ge kd$, whence $d'\geq kd$. 
Write
\[
w =
\begin{bmatrix}
w_1 \\
w_{2} \\ 
\vdots \\
w_{k+1}
\end{bmatrix},
\] 
where each $w_i\in \F^{n}$. Then by the definition of the $a_j'$ in Construction~\ref{conadm}, 
\begin{equation}\label{eqwj}
w_i=\sum_{j=1}^k c_{i-j} a_j, \quad \mbox{for $i=1,\dots,k+1$, reading subscripts modulo $k+1$}.
\end{equation}
In particular each $w_i$ is a codeword of $C$. Thus if each of the $w_i$ is nonzero, then 
$\wt(w) = \sum_{i=1}^{k+1}\wt(w_i) \geq (k+1)d$. 
Suppose now that, for some $i$, $w_i=0$. Since the $a_j$ are linearly independent, it follows from \eqref{eqwj} that $c_{i-j}=0$ for $j=1,\dots,k$. Then since $w\ne 0$, we must have $c_{i}\ne 0$ so $w=c_ia_i'$, and hence
\begin{equation}\label{conadm2}
\wt(w) = \wt(a_{i}') =\sum_{j=1}^{k}\wt(a_j) \geq kd    
\end{equation}
since each $a_i$ is a codeword of $C$ and so has weight at least $d$. This proves part (a).

(b) Now assume that $C$ is $u$-bounded relative to $B$. Let $w=\sum_{j=1}^{k+1} c_ja_j'$ be an arbitrary non-zero codeword in $\code(B)$, and define the $w_i$ as above so that $\wt(w) =\sum_{i=1}^{k+1} \wt(w_i)$, and equation \eqref{eqwj} holds. If each of the $w_i$ is nonzero then  $\wt(w)\ge (k+1)d$. On the other hand if some $w_i=0$, then we showed above that equation \eqref{conadm2} holds, so $\wt(w) =\sum_{j=1}^{k} \wt(a_j)$. In this case, by Definition~\ref{defadm}(i), we have $\wt(a_j)=u$ for each $j$, and hence $\wt(w) = ku$. Further, by Definition~\ref{defadm}(iii), $u\ge d(1+k^{-1})$, and hence $\wt(w)=ku\geq (k+1)d$. Thus  the minimum distance $d'$ of $\code(B)$ satisfies $d'\geq(k+1)d$. To see that equality holds, consider the 
the codeword $w$ obtained by taking $c_1=\dots=c_{k+1}=1$. For each $i$, equation \eqref{eqwj} shows that $w_i=\sum_{j=1}^ka_j$, and hence $\wt(w_i)=d$ by Definition~\ref{defadm}(ii). Therefore, for this codeword $w$ we have $\wt(w)=\sum_{i=1}^{k+1} \wt(w_i) = (k+1)d$, and we conclude that $d'=(k+1)d$.   

Now set $u'=ku$. First, from the definition of the $a_j'$ we have $\wt(a_j') = \sum_{i=1}^k\wt(a_i)$, which is equal to $ku=u'$ by  Definition~\ref{defadm}(i). 
Second, we showed in the previous paragraph that $\wt(\sum_{i=1}^{k+1} a_i') =(k+1)d=d'$. Third, since $u'=ku$ and $d'(1+(k')^{-1}) = (k+1)d(1+(k+1)^{-1}) = d(k+2)$, it follows that $u'\geq d'(1+(k')^{-1})$ if and only if $u\geq d(k+2)/k=d(1+2k^{-1})$. Thus, by Definition~\ref{conadm},  $\code(B)$ is $u'$-bounded relative to $\basis(B)$ if and only if $u\ge d(1+2k^{-1})$. This completes the proof. 
\end{proof}
\subsection{Recursive applications of the construction}

By Proposition~\ref{propadm}(a),  the output code $\code(B)$ of Construction~\ref{conadm}  is always a linear code with ordered basis $\basis(B)$, and hence  Construction~\ref{conadm} may be applied repeatedly, producing larger and larger codes.  Moreover, Proposition~\ref{propadm}(b) implies that, provided the parameter $u$ is 
sufficiently large, $\code(B)$ is bounded relative to $\basis(B)$.  We investigate how many times we may apply Construction~\ref{conadm}  recursively and still obtain a bounded code.  To keep track of these repeated applications of Construction~\ref{conadm} we introduce some natural notation for the output codes and bases.
\begin{definition}\label{defadm2}{\rm
Let  $i$ be a positive integer, and let $C$ be an $\F$-linear code with an ordered basis $B$. In terms of the output of Construction~\ref{conadm} applied repeatedly to $C, B$, for $i=1$, let
\[
\code(B,1)=\code(B) \quad\text{and}\quad \basis(B,1)=\basis(B),
\]
and for $i\geq2$, define recursively,
\[
    \code(B,i) =\code(\basis(B,i-1))\quad\text{and}\quad \basis(B,i)=\basis(\basis(B,i-1)). 
\]
}
\end{definition}

We examine the parameters and boundedness of $\code(B,i)$ for various values of $i$.

\begin{proposition}\label{propadm2}
Let  $i$ be a positive integer, and let $C$ be an $\F$-linear $[n,k,d]$-code  with an ordered basis $B$. 
\begin{enumerate}[{\rm (a)}]
\item Then $\code(B,i)$ is an $\F$-linear $[n_i,k_i,d_i]$-code, where 
\[
n_i= n\prod_{j=1}^{i}(k+j), \quad k_i=k+i,\quad d_i\geq d\prod_{j=1}^{i}(k+j-1).
\]
\item Further, suppose that $C$ is $u$-bounded relative to $B$, and $u\ge d(1+ik^{-1})$, that is $i\le k(ud^{-1}-1)$. Then the minimum distance $d_i$ of $\code(B,i)$ is
\[
d_i=d\prod_{j=1}^{i}(k+j),\quad \mbox{and setting}\quad u_i=u\prod_{j=1}^{i}(k+j-1),
\]
$\code(B,i)$ is $u_i$-bounded relative to $\basis(B,i)$ if and only if $u\ge d(1+(i+1)k^{-1})$.
\end{enumerate}
\end{proposition}

\begin{proof}
(a) Our proof is by induction on $i$. It follows from Proposition~\ref{propadm}(a) that $\code(B,1)$ is an $\F$-linear $[n_1,k_1,d_1]$-code
where $n_1=n(k+1), k_1=k+1$ and $d_1\geq dk$. Suppose that $i\geq2$ and assume that the assertions of part~(a) hold for $\code(B,i-1)$.
Proposition~\ref{propadm}(a) implies that $n_i=n_{i-1}(k_{i-1}+1), k_i=k_{i-1}+1$ and $d_i\geq d_{i-1}k_{i-1}$. Hence the assertions of part~(a) hold for also for $\code(B,i)$.

(b) Now suppose that $C$ is $u$-bounded relative to the ordered basis $B$. Then $u\geq d(1+k^{-1})$ by Definition~\ref{defadm}(iii), and by Proposition~\ref{propadm}(b), the minimum distance $d_1$ of $\code(B,1)$ is $d_1=d(k+1)$; also  $\code(B,1)$ is $u_1$-bounded, where $u_1=ku$, if and only if $u\geq d(1+2k^{-1})$. Thus part (b) holds for $i=1$.  Assume now that $i\geq2$ and $i\leq k(ud^{-1}-1)$, and also that part (b) holds for $\code(B,i-1)$, that is, 
\[
 d_{i-1}=d\prod_{j=1}^{i-1}(k+j),\quad \mbox{and setting}\quad u_{i-1}=u\prod_{j=1}^{i-1}(k+j-1),
\]
$\code(B,i-1)$ is $u_{i-1}$-bounded relative to $\basis(B,i-1)$ if and only if $u\ge d(1+ik^{-1})$. Since we are assuming that the inequality 
$u\ge d(1+ik^{-1})$ holds, we conclude that $\code(B,i-1)$ is $u_{i-1}$-bounded relative to $\basis(B,i-1)$. Thus applying Proposition~\ref{propadm}(b) to  $\code(B,i-1)$ with  $\basis(B,i-1)$, shows that $d_i=d_{i-1}(k_{i-1}+1)$ and that $\code(B,i)$ is $(u_{i-1}k_{i-1})$-bounded relative to $\basis(B,i)$ if and only if $u_{i-1}\geq d_{i-1}(1+2k_{i-1}^{-1})$.
Since $k_{i-1}+1=k+i$ by part (a), we conclude that $d_i$ is as in part (b).  Also, since $u_{i-1}k_{i-1}=u\prod_{j=1}^{i}(k+j-1)$ is equal to $u_i$, $\code(B,i)$ is $u_i$-bounded relative to $\basis(B,i)$ if and only if the inequality $u_{i-1}\geq d_{i-1}(1+2k_{i-1}^{-1})$ holds. Substituting the values above for $u_{i-1}, d_{i-1}, k_{i-1}$, this inequality is 
\[
u\prod_{j=1}^{i-1}(k+j-1) \geq \left(d\prod_{j=1}^{i-1}(k+j) \right) \frac{k+i+1}{k+i-1},
\]
that is, $u\geq dk^{-1}(k+i+1)=d(1+(i+1)k^{-1})$. Thus part (b) is proved by induction. 
\end{proof}

\begin{remark}
We note that, the quantity mentioned in the introduction, namely the rate times the minimum distance, grows slowly with the number $i$ of applications of Construction~\ref{conadm}: for $1\le i \le k(ud^{-1}-1)$ we have, by Proposition~\ref{propadm2}, that 
\[
\frac{k_id_i}{n_i}=\frac{(k+i)d}{n} = \frac{kd}{n} + \frac{id}{n}.
\]
\end{remark}

\section{Explicit instances of the construction}
In this section we construct a family of linear codes which are defined as column spaces of a certain family of square matrices. First we introduce the matrices.

\subsection{A family of matrices}\label{sub:mx}
As always $\F$ denotes an arbitrary field. We introduce an infinite family of matrices over $\F$, and explore some of their properties. We denote the space of $m\times n$ matrices over $\F$ by $\F^{m\times n}$; we denote the zero matrix in this ring by $0_{m\times n}$, or sometimes just by $0$; and if $m=n$ we denote the identity matrix by $I_m$, and the determinant of $\A\in\F^{m\times m}$ by $\det(\A)$. For a matrix $\A\in\F^{m\times n}$, its $ij$-entry is denoted $\A_{ij}$, and its \emph{transpose} obtained by interchanging rows and columns is denoted $\A^T$, that is $(\A^T)_{ij}=\A_{ji}$ for all $i,j$. For simplicity of the exposition we will sometimes represent matrices as block matrices where some $(a\times b$)-block may be an \emph{empty matrix}, that is $a=0$ or $b=0$ are allowed.


\begin{definition}\label{deff}{\rm
For a positive integer $i$, we define matrices $\A_i$ and $\B_i$ as follows. For $i=1$, these are
\[
\A_1=
\begin{bmatrix}
 0 & -1\\
 1 &  0
\end{bmatrix}
,\quad\text{and}\quad
\B_1=
\begin{bmatrix}
 1 & -1\\
 -1 &  1
\end{bmatrix},
\]
and recursively, for $i\geq1$, if $\A_i, \B_i$ have been defined, then we let 
\[
\A_{i+1}=
\begin{bmatrix}
 \A_1 & \B_i\\
 -\B_i^T &  \A_i
\end{bmatrix},
\quad \text{and}\quad
\B_{i+1}=
\begin{bmatrix}
 \B_1 & \B_i
\end{bmatrix}.
\]
}
\end{definition}
Our first observations about the family are given in this lemma. The idea behind the proof of Lemma~\ref{lem:mx1} we found in \cite{Silvester}*{Section 5, Acknowledgements}. 

\begin{lemma}\label{lem:mx1}
Let $i$ be a positive integer and $\A_i, \B_i$ be as in Definition~$\ref{deff}$. Then
\begin{enumerate}[{\rm (a)}]
    \item $\B_i=[ \B_1\cdots \B_1]\in\F^{2\times 2i}$; for $a\in\{1,2\}$ and $b\in\{1,\dots,2i\}$ we have $(\B_i)_{ab}=(-1)^{a+b}$ and $\B_i^T\A_1\B_i=\B_i^T\A_1^{-1}\B_i=0_{2\times 2}$.
    \item $\A_i\in\F^{2i\times 2i}$ satisfies $\det(\A_i)=1$.
\end{enumerate}
\end{lemma}

\begin{proof}
(a) A straightforward inductive argument, based on the definition of $\B_{i+1}$ in Definition~\ref{deff} shows that $\B_i=[ \B_1\cdots \B_1]\in\F^{2\times 2i}$ for each $i\ge1$. Then it is easy to see that $(\B_i)_{ab}=(-1)^{a+b}$ for all $a, b$. It follows from $\B_i=[ \B_1\cdots \B_1]$ that $\B_i^T\A_1\B_i= \sum_{j=1}^i \B_1^T\A_1\B_1$, and computing we see that  $\B_1^T\A_1\B_1 = \B_1\A_1\B_1 = 0_{2\times 2}$. Hence $\B_i^T\A_1\B_i=0_{2\times 2}$, and the final assertion $\B_i^T\A_1^{-1}\B_i=0_{2\times 2}$ follows from the fact that $\A_1^{-1}=-\A_1$. 

(b) It is easy to see that $\det(\A_1)=1$, and $\A_{i}\in\F^{2i\times 2i}$ for each $i\ge1$. Assume that $i\geq 1$, assume inductively that $\det(\A_i)=1$, and consider $\A_{i+1}$ as defined in Definition~\ref{deff}.  Modifying an idea from  \cite{Silvester}*{Section 5}, we see by direct computation that 
\[
\begin{bmatrix}
 \A_1 & \B_i\\
 -\B_i^T &  \A_i
\end{bmatrix}
\begin{bmatrix}
 I_2 & -\A_1^{-1}\B_i\\
 0_{2i\times 2} &  I_{2i}
\end{bmatrix}
=
\begin{bmatrix}
 \A_1 & 0_{2\times 2i}\\
 -\B_i^T & \B_i^T\A_1^{-1}\B_i+ \A_i
\end{bmatrix}
=
\begin{bmatrix}
 \A_1 & 0_{2\times 2i}\\
 -\B_i^T & \A_i
\end{bmatrix}
\]
where the last equality follows from part (a). Since the second matrix on the left side has determinant 1, it follows that $\det(\A_{i+1})=\det(\A_1)\det(\A_i)=1$ where the last equality uses the inductive hypothesis. Thus part (b) follows by induction.
\end{proof}

The next lemma looks at properties of the columns of the $\A_i$.

\begin{lemma}\label{lem:mx2}
Let $i$ be a positive integer, and let $\A_{i}=[a_1,a_2,\dots, a_{2i}]$ 
be as in Definition~$\ref{deff}$, where $a_j$ is the $j$th column of $\A_i$. Then
\begin{enumerate}[{\rm (a)}]
    \item $\wt(a_j)=2i-1$ for $j\in\{1,\dots,2i\}$,
    \item $a_{2j-1}+a_{2j} = [0_{1\times 2(j-1)},-1,1,0_{1\times 2(i-j)}]^T \in\F^{2i\times 1}$
    for $j\in\{1,\dots,i\}$,
    \item $\sum_{j=1}^{2i-1} a_j = [0,\dots,0,1]^T\in\F^{2i\times 1}$, and in particular, $\wt(\sum_{j=1}^{2i-1} a_j) = 1$.
\end{enumerate}
\end{lemma}

\begin{proof}
(a) Part (a) follows immediately from Definition~\ref{deff}.

(b) We use induction on $i$. For $i=1$, $a_1+a_2 = [-1,1]^T$ by  Definition~\ref{deff}, so (b) holds in this case (as $i=j=1$ and the first and last entries in this vector in (b) are empty matrices). Now let $i\geq 1$ and assume inductively that (b) holds for $\A_i=[a_1,\dots,a_{2i}]$, and consider $\A_{i+1} = [a_1',\dots,a_{2i+2}']$. By Definition~\ref{deff} and the structure of $\B_i$ given by Lemma~\ref{lem:mx1}(a), it follows that $a_1'+a_2' = [-1, 1, 0_{1\times 2i}]^T\in\F^{2(i+1)\times 1}$. For $j\in\{2,\dots,i+1\}$, it follows from Definition~\ref{deff} that 
\[
a_{2j-1}'+a_{2j}' = 
\begin{bmatrix}
 0_{2\times 1}\\
 a_{2j-3}+a_{2j-2}
\end{bmatrix}.
\]
Hence the structure of $a_{2j-1}'+a_{2j}'$ follows from the inductive hypothesis. Thus part~(b) is proved by induction.

(c) This part is also proved by induction on $i$. The case $i=1$ follows from the definition of $\A_1$. So assume that $i\geq1$ and that part~(c) holds for $\A_i$, that is $b:=\sum_{\ell=1}^{2i-1} a_\ell = [0,\dots,0,1]^T\in\F^{2i\times 1}$. Suppose that $\A_i$ and $\A_{i+1}$ have columns $a_j$ and $a'_j$ as in the proof of part (b). Using the structure of $\B_i$ from Lemma~\ref{lem:mx1}(a), the form of $a_1'+a_2'$ above, and the inductive hypothesis, we have
\[
\sum_{j=1}^{2i+1}a_{j}' = (a_1'+a_2') + 
\begin{bmatrix}
 1\\
 -1\\
 \sum_{\ell=1}^{2i-1} a_\ell
\end{bmatrix}
= 
\begin{bmatrix}
 -1\\
 1\\
 0_{2i\times 1}
\end{bmatrix}
+ \begin{bmatrix}
 1\\
 -1\\
b
\end{bmatrix}
= [\ 0,\ \dots\ ,\ 0,\ 1\ ]^T\in\F^{2(i+1)\times 1}.
\]
Part (c) follows by induction.
\end{proof}
\subsection{Codes constructed from matrices}

We now construct $\F$-linear codes from the matrices in Subsection~\ref{sub:mx}, noting that the columns of $\A_i$ are linearly independent by Lemma~\ref{lem:mx1}(b) . 

\begin{definition}\label{defff}{\rm
Let $i$ be a positive integer, let  $\A_{i}=[a_1,a_2,\dots,a_{2i}] \in\F^{2i\times 2i}$ be
as in Definition~\ref{deff}, where $a_j$ is the $j$th column of $\A_i$, and let 
$C(i)$ be the $\F$-linear code in $\F^{2i}$ with ordered basis $B(i)=(a_1,\dots,a_{2i-1})$.
}
\end{definition}
We use the results from Subsection~\ref{sub:mx} to determine the properties of $C(i)$, and the results from Section~\ref{sec:bdd} to study the codes obtained using Proposition~\ref{propadm2}.

\begin{proposition}\label{propf}
For a positive integer $i$, let $C(i)$ and $B(i)$ be as in Definition~$\ref{defff}$. 
\begin{enumerate}[{\rm (a)}]
\item
For  $i \ge 2$, the code $C(i)$ is an $\F$-linear $[2i,2i-1,1]$-code which is $(2i-1)$-bounded relative to the ordered basis $B(i)$.

\item
For $1\leq j\leq 4i^2-6i+1$, the code $\code(B(i),j)$, as defined in Definition~$\ref{defadm2}$,  is  an $\F$-linear $[n_j,k_j,d_j]$-code that is $u_j$-bounded relative to the ordered basis $\basis(B(i),j)$, where 
 \[
n_j=2i\prod_{\ell=1}^{j}(2i-1+\ell), \quad k_j=2i-1+j, \quad d_j=\prod_{\ell=1}^{j}(2i-1+\ell) , \ \mbox{and}\ u_j=(2i-1)\prod_{\ell=1}^{j}(2i-2+\ell).
\]
\end{enumerate}
\end{proposition}
\begin{proof}
(a)
By Lemma ~\ref{lem:mx1}(b), $B(i)$ is an ordered basis for $C(i)$. Thus $C(i)$ has length $2i$ and dimension is $2i-1$. Moreover, $d(C(i))=1$ as $[0,\dots,0,1]^T\in C(i)$ by Lemma~\ref{lem:mx2}(c).
By Lemma~\ref{lem:mx2}(a), $\wt(a_j)=2i-1$ for each $j\leq 2i-1$, and $\wt(\sum_{j=1}^{2i-1}a_j)=1$. 
To see that $C(i)$ is $(2i-1)$-bounded relative to $B(i)$ for each $i\geq2$ we use Definition~\ref{defadm}
and note that $2i-1 \ge 1+(2i-1)^{-1}$ holds (this is true for $i\geq2$).  This proves part~(a).

(b) We apply Proposition~\ref{propadm2}(b) to $C(i)$ and $B(i)$. By part (a) the parameters $n, u, k, d$ of that result are: $n=2i, u=k=2i-1, d=1$, and we wish to apply  Proposition~\ref{propadm2}(b) with a positive integer parameter $j$ satisfying $u\geq d(1+(j+1)k^{-1})$, or equivalently $j+1\leq k(ud^{-1}-1) = 
(2i-1)(2i-2)=4i^2-6i+2$. This latter inequality is our assumption on $j$, so we may   
therefore apply Construction~\ref{conadm}  $j$ times and conclude,  by  Proposition~\ref{propadm2}(b), that  $\code(B(i),j)$ is  an $\F$-linear $[n_j,k_j,d_j]$-code that is $u_j$-bounded relative to the ordered basis $\basis(B(i),j)$, where 
the parameters  $n_j,k_j,d_j, u_j$ are as in the statement. This proves part (b).
\end{proof}
In Proposition~\ref{propf}(b) we have $n_j=2i d_j$ and $u_j(2i+j-1)=(2i-1)d_j$, respectively.
Note that setting $j=0$ in the expressions for $n_j, k_j, d_j, u_j$ in Proposition~\ref{propf}(b), we obtain the parameters $n,k,d,u$, respectively, for the code $C(i)$ examined in Proposition~\ref{propf}(a). Thus if we set $\code(B(i),0):=C(i)$, then all the assertions of Proposition~\ref{propf}(b) hold with $j=0$. Thus we   
have constructed explicitly a two-parameter family of bounded linear codes over an arbitrary field $\F$, namely:
\[
\text{$\code(B(i),j)$ for all $i\geq 2$ and all integers $j$ satisfying $0\leq j\leq 4i^2-6i+1$.}
\]
We write $\code_{\,\F}(B(i),j)$ and $\basis_{\,\F}(B(i),j)$ if we wish to emphasize the field $\F$ of scalars. To prove our main result we investigate a one-parameter subfamily of these codes,  choosing,  for each value of $i$, the code with the largest value of $j$. We now replace $i$ with~$i+1$.

\begin{definition}\label{def:Fi}{\rm
For a field $\F$ of scalars, and each $i\geq 1$, define
\[
C_i := \code_{\,\F}(B(i+1), 4i^2-2i-1)\quad\textup{and}\quad \B_i:= \basis_{\,\F}(B(i+1), 4i^2-2i-1).
\]
}

\end{definition} 
The following theorem follows immediately from Proposition~\ref{propf}.
\begin{theorem}\label{corf}
Given a field $\F$  and an integer $i \ge 1$, the code $C_i$ is a linear $[N_i,K_i,D_i]$-code over $\F$, which is $U_i$-bounded relative to its ordered basis $\B_i$, where 
\[
N_i=2(i+1)\prod_{\ell=1}^{4i^2-2i-1}(2i+1+\ell), \quad K_i=4(i+1)i, \quad D_i=\prod_{\ell=1}^{4i^2-2i-1}(2i+1+\ell) , \ 
\]
and $U_i=(2i+1)\displaystyle\prod_{\ell=1}^{4i^2-2i-1}(2i+\ell)$.
Moreover, $\displaystyle{\frac{K_iD_i}{N_i}= 2i}$ grows linearly with $i$. 
\end{theorem}

\begin{proof}
The first assertions follow immediately from Proposition~\ref{propf} by replacing $i$ with $i+1$. Finally, from these parameters we see that $\displaystyle{\frac{K_iD_i}{N_i}= \frac{4(i+1)i}{2i}= 2i}$.
\end{proof}

\begin{proof}[Proof of Theorem~\ref{main}]
  By Proposition~\ref{corf} there are $[n_i,k_i,d_i]$-codes $C_i$, $i\ge1$, where $k_id_i/n_i=2i$ and $k_i=4(i+1)i$. A simple calculation shows that
  \[2i>\sqrt{4i(i+1)}-1>\sqrt{4i^2}-1=2i-1.\qedhere\]
\end{proof}

\section*{Acknowledgement}
The authors thank Patrick Sol\'e for his remarks and for alerting us to~\cite{ShiWuSole}. 
The first author is grateful to UWA and the Centre for the Mathematics of Symmetry and Computation, and also to Shahid Rajaee Teacher Training University for their financial support when she was a visiting scholar. The second and third authors acknowledge support from the Australian Research Council Discovery Project DP190100450.

\end{document}